\documentclass[11pt]{article}
\usepackage[margin=1.in]{geometry}

\usepackage{ifthen}

\newboolean{lncs}   
\setboolean{lncs}{false}

\newboolean{linenumbers}   
\setboolean{linenumbers}{false}

\ifthenelse{\boolean{lncs}}{
\usepackage{amsfonts}
}{
  \usepackage{amsmath,amssymb} 
  \usepackage{amsthm}
}

\usepackage[utf8]{inputenc}
\usepackage{graphicx,dblfloatfix}
\usepackage{subfig}
\usepackage{url}

\usepackage{cite}
\usepackage{mathtools}
\usepackage[ruled,vlined,linesnumbered]{algorithm2e}
\usepackage{paralist}
\usepackage{todonotes}
\usepackage{thmtools, thm-restate}
\usepackage{xcolor}

\usepackage{paralist} 

\usepackage[shortlabels]{enumitem}

\usepackage{xifthen}
\newcommand{\ifthen}[2]{\ifthenelse{#1}{#2}{}}
\newcommand{\ifnot}[2]{\ifthenelse{#1}{}{#2}}

\usepackage{hyperref}

\usepackage[capitalize]{cleveref}
\crefname{theorem}{Theorem}{theorems}
\crefname{definition}{Definition}{definitions}
\crefname{proposition}{Proposition}{propositions}
\crefname{lemma}{Lemma}{lemmas}
\crefname{corollary}{Corollary}{corollaries}
\crefname{claim}{Claim}{claims}
\crefname{observation}{Observation}{observations}
\crefname{fact}{Fact}{facts}
\crefname{dfn}{Definition}{definitions}
\crefname{obs}{Observation}{observations}
\crefname{pb}{Problem}{problems}
\crefname{algo}{Algorithm}{algorithms}

\newcounter{maincounter} 
\ifthenelse{\boolean{lncs}}{
}{
  \newtheorem{theorem}[maincounter]{Theorem}

\newtheorem{proposition}[maincounter]{Proposition}
\newtheorem{lemma}[maincounter]{Lemma} 
\newtheorem{corollary}[maincounter]{Corollary}
\newtheorem{claim}[maincounter]{Claim}

}

\newtheorem{fact}[maincounter]{Fact}

\ifthenelse{\boolean{linenumbers}}{ 
\usepackage[left, pagewise]{lineno}
\linenumbers
\newcommand*\patchAmsMathEnvironmentForLineno[1]{%
  \expandafter\let\csname old#1\expandafter\endcsname\csname #1\endcsname
  \expandafter\let\csname oldend#1\expandafter\endcsname\csname end#1\endcsname
  \renewenvironment{#1}%
     {\linenomath\csname old#1\endcsname}%
     {\csname oldend#1\endcsname\endlinenomath}}%
\newcommand*\patchBothAmsMathEnvironmentsForLineno[1]{%
  \patchAmsMathEnvironmentForLineno{#1}%
  \patchAmsMathEnvironmentForLineno{#1*}}%
\AtBeginDocument{%
\patchBothAmsMathEnvironmentsForLineno{equation}%
\patchBothAmsMathEnvironmentsForLineno{align}%
\patchBothAmsMathEnvironmentsForLineno{flalign}%
\patchBothAmsMathEnvironmentsForLineno{alignat}%
\patchBothAmsMathEnvironmentsForLineno{gather}%
\patchBothAmsMathEnvironmentsForLineno{multline}%
}
}{
}

\makeatletter
\newdimen\commentwd
\let\oldtcp\tcp
\def\alignedtcp*[#1]#2{
\setbox0\hbox{#2}%
\ifdim\wd\z@>\commentwd\global\commentwd\wd\z@\fi
\oldtcp*[r]{\leavevmode\hbox to \commentwd{\box0\hfill}}}

\let\oldalgorithm\algorithm
\def\algorithm{\oldalgorithm
\global\commentwd\z@
\expandafter\ifx\csname commentwd@\romannumeral\csname c@\algocf@float\endcsname\endcsname\relax\else
\global\commentwd\csname commentwd@\romannumeral\csname c@\algocf@float\endcsname\endcsname
\fi
}
\let\oldendalgorithm\endalgorithm
\def\endalgorithm{\oldendalgorithm
\immediate\write\@auxout{\gdef\expandafter\string\csname commentwd@\romannumeral\csname c@\algocf@float\endcsname\endcsname{%
\the\commentwd}}}

\DeclareMathOperator{\dist}{dist} 

\renewcommand{\O}{\ensuremath{\mathcal{O}}}
\newcommand{\noise}{\ensuremath{\operatorname{N}}}

\newcommand{\E}{\ensuremath{\mathcal{E}}}
\renewcommand{\P}{\ensuremath{\mathbb{P}}}

\usepackage{placeins}

\newcommand{\RandCDeltaColoring}{{\sc Rand4DeltaColoring}}
\newcommand{\ColorReduction}{{\sc ColorReduction}}

\newcommand{\etc}{\ensuremath{\E^t_\text{confl}}}
\newcommand{\ptc}{\ensuremath{\P^t_\text{confl}}}
\newcommand{\pc}{\ensuremath{\P_\text{confl}}}
\newcommand{\ets}{\ensuremath{\E^t_\text{succ}}}
\newcommand{\pts}{\ensuremath{\P^t_\text{succ}}}
\newcommand{\etf}{\ensuremath{\E^t_\text{fail}}}
\newcommand{\ptf}{\ensuremath{\P^t_\text{fail}}}

\newcommand{\blockline}{\noindent\hspace{-0.025\textwidth}%
    \rule[8.5pt]{0.989\linewidth}{0.8pt} \\[-0.80\baselineskip] }
\newcounter{algo} 
\newcounter{subalgo}[algo] 
\renewcommand{\thesubalgo}{\thealgo.\arabic{subalgo}}

\SetCommentSty{mycommfont}

\begin{document}

\title{Simple~Distributed~$\Delta+1$~Coloring~in~the~SINR~Model}

\author{
  Fabian Fuchs, 
 Roman Prutkin \\
  {\normalsize Karlsruhe Institute for Technology } \\
  {\normalsize   Karlsruhe, Germany } \\
  {\normalsize   \{fabian.fuchs, roman.prutkin\}@kit.edu } 
}

\maketitle



\addtolength{\parskip}{-1pt}

\begin{abstract}
  In wireless ad hoc or sensor networks, distributed node coloring is
  a fundamental problem closely related to establishing efficient
  communication through TDMA schedules. For networks with maximum
  degree~$\Delta$, a~$\Delta+1$ coloring is the ultimate goal in the
  distributed setting as this is always possible. In this work we
  propose~$\Delta+1$ coloring algorithms for the synchronous and
  asynchronous setting. All algorithms have a runtime of
 ~$\O(\Delta \log n)$ time slots. This improves on the previous
  algorithms for the SINR model either in terms of the number of
  required colors or the runtime and matches the runtime of local
  broadcasting in the SINR model (which can be seen as a lower bound).
\end{abstract}

\section{Introduction}
\label{sec:introduction}

One of the most fundamental problems in wireless ad hoc or sensor
networks is efficient communication. Indeed, most algorithms concerned
with the physical or \emph{Signal-to-Interference-and-Noise-Ratio}
(SINR) model consider algorithms to establish initial communication
right after the network begins to operate. However, those initial
methods of communication are not very efficient, as there are either
frequent collisions and reception failures due to interference, or
time is wasted in order to provably avoid such collisions and
failures. In the case of local
broadcasting\cite{gmw-lbpim-08,hm-ttblb-12,yhwl-aodaa-12,fw-lbatp-14},
a multiplicative~$\O(\Delta \log n)$ factor is required to execute
message-passing algorithms in the SINR model, where~$\Delta$ is the
maximum degree in the network, cf. Section~\ref{sec:prelim-coloring}
for a definition. Thus, wireless networks often use a more refined
transmission schedule as part of the Medium Access Control (MAC)
layer. One of the most popular solutions to the medium access problem
are Time-Division-Multiple-Access (TDMA) schedules, which provide
efficient communication by assigning nodes to time slots. The main
problem in establishing a TDMA\footnote{This also holds for related
  techniques such as Frequency-Division-Multiple-Access (FDMA).}
schedule can be reduced to a distributed node coloring. Given a node
coloring, we can establish a transmission schedule by simply
associating each color with one time slot. The node coloring
considered in this work ensure that two nodes capable of communicating
directly with each other do not select the same color.  Note that a
TDMA schedule based on such a coloring is not yet feasible in the SINR
model. However, a TDMA schedule that is feasible in the SINR model can
be computed based on our coloring, for example as shown in
\cite{dt-dncsi-10,fw-olbsc-13}.

The problem of distributed node coloring dates back to the early days
of distributed computing in the mid-1980s. In contrast to centralized
node coloring, a~$\Delta + 1$ coloring is considered to be the
ultimate goal in distributed node coloring as it is already
NP-complete to compute the chromatic number (i.e., the minimum number
of colors required to color the graph) in the centralized
setting\cite{gj-np-book-79}.  There is a rich line of research in this
area, however, most of the work has been done for message-passing
models like the~$\mathcal{LOCAL}$ model. Such models are designed for
wired networks and do not fit the specifics of wireless networks.

In the SINR model, wireless communication is modelled based on the
signal transmission and a geometric decay of the signal strength. In
contrast to graph-based models such as the protocol model, which model
interference between nodes as a purely local property, the SINR model
accounts for both local and global interference. Furthermore, it has
been shown that the protocol model is quite limited in modeling
reality\cite{mww-pdbgbm-06}. Often, the SINR model is denoted as the
physical model due to its common use in electrical engineering. Thus,
the algorithms designed for the SINR model are far more realistic and
fit the specifics of wireless networks better.

In this work, we use two simple and very well-known algorithms
(covered for example in \cite{be-dcg-13}) designed for message-passing
models, and show that we can \emph{efficiently} execute the algorithms
in the SINR model. However, this cannot be achieved by a simple
simulation of each round of the message passing algorithm by one
execution of local broadcasting as this results in a runtime of
$\O(\Delta \log^2 n)$ time slots. Instead, we modify both the
communication rounds in the SINR model and the algorithms to perfectly
fit together. The synergy effect of our careful adjustments is that
the coloring algorithm runs in~$\O(\Delta \log n)$ time slots, which
is asymptotically exactly the runtime of one local
broadcast\cite{gmw-lbpim-08}. This matches the runtime of current
$\O(\Delta)$ coloring algorithms\cite{dt-dncsi-10}, and improves on
current~$\Delta+1$ coloring algorithms for the SINR model which
require~$\O(\Delta \log n + \log^2 n)$ or~$\O(\Delta \log^2 n)$ time
slots\cite{ywhl-ddcpm-14-journal}.

Note that our runtime cannot be achieved using algorithms in the
$\mathcal{LOCAL}$ model, as state-of-the-art algorithms for
distributed node coloring in the~$\mathcal{LOCAL}$ model run in
$\O(\log \Delta + 2^{\O(\sqrt{\log \log n})})$
rounds\cite{beps-tldsb-12} (in general graphs) or~$\O(\log^* n)$
rounds\cite{sw-alsdmis-08} (in growth bounded graphs) of local
broadcasting. Indeed, a lower bound of~$\frac{1}{2} \log^*n$ for
distributed node coloring in the~$\mathcal{LOCAL}$ model due to Linial
\cite{l-dgagl-87} exists. Thus, in contrast to the algorithm proposed
in this work, the number of time slots required to execute distributed
node coloring algorithms designed for the~$\mathcal{LOCAL}$ model in
the SINR model must require~$\Omega(\log^*n)$ local broadcasting~executions.

The communication between nodes in our algorithm is based on the local
broadcasting algorithm proposed by Goussevskaia
\emph{et. al.}\cite{gmw-lbpim-08}.  Thus, we require the nodes to know
an upper bound on the maximum number of nodes in a node's surroundings
(which we call proximity area, cf. Section~\ref{sec:prelim-coloring}),
an upper bound on the number of nodes in the network, as well as some
model-related hardware constants in order to enable initial
communication.  Such requirements are common in the SINR model,
however, we discuss the case that the number of nodes in the proximity
area is not known in \cref{sec:without-knowledge-delta}. We
discuss related work and our contributions in the next section.

\subsection{Related Work and Contributions}
\label{sec:related-work}

Due to the rich amount of work on distributed node coloring in the
message-passing model, we shall only highlight the most efficient
algorithms for~$\Delta+1$ colorings here. For a more thorough overview
we refer to a recent monograph by Barenboim and Elkin
\cite{be-dcg-13}. 
The fastest deterministic algorithm for general graphs is due to
Awerbuch \emph{et. al.}~\cite{algp-ndldc-89} and Panconesi and
Srinivasan~\cite{ps-ocdnd-96} and runs in~$2^{\O(\sqrt{\log n})}$.
For moderate values of~$\Delta$, another deterministic algorithm with
runtime~$\O(\Delta) + \frac{1}{2} \log^*n$ is due to Barenboim, Elkin
and Kuhn~\cite{bek-ddclt-14}. For growth bounded graphs, which
generalize unit disk graphs, a deterministic distributed algorithm due
to Schneider and Wattenhofer~\cite{sw-alsdmis-08} computes a
valid~$\Delta+1$ coloring in~$\O(\log^* n)$ rounds.
Regarding randomized algorithms, the algorithm on which this work is
based was the most efficient~$\Delta+1$ coloring algorithm until
recently~\cite[Chapter~10]{be-dcg-13}. The algorithm can be seen as a
simple variant of Luby's maximal independent set
algorithm~\cite{l-rrpcp-88}.  In recent years the problem received
considerable attention which culminated in the currently best
randomized algorithm running
in~$\O(\log \Delta + 2^{\O(\sqrt{\log \log n})})$ due to Barenboim
\emph{et. al.}~\cite{beps-tldsb-12}.

In wireless networks, the SINR model received increasing attention
first in the electrical engineering community, and was picked up by
the algorithms community due to a seminal work by Gupta and
Kumar~\cite{gk-cwn-00}.  An overview of works regarding transmission
scheduling in the SINR model can be found in a survey by Goussevskaia,
Pignolet and Wattenhofer~\cite{Goussevskaia2010Efficiency}. Other
results cover Broadcasting \cite{dgkn-bahs-13,jkrs-drbwn-13}, Local
Broadcasting\cite{gmw-lbpim-08,hm-ttblb-12,yhwl-aodaa-12,fw-lbatp-14}
and backbone construction\cite{jk-dbsas-12}.  Regarding distributed
node coloring in the SINR model Derbel and Talbi \cite{dt-dncsi-10}
show that a distributed node coloring algorithm due to Moscibroda and
Wattenhofer \cite{mw-curn-08,mw-curn-05} can be adapted to the SINR
model. They provide an algorithm that computes an~$\O(\Delta)$
coloring in~$\O(\Delta \log n)$ time slots. The algorithm first
computes a set of leaders using a maximal independent set (MIS,
cf. \cref{sec:prelim-coloring}) algorithm, then leader nodes assign
colors to non-leaders, which again compete for their final color with
a restricted number of neighboring nodes that may have received the
same assignment.  Yu
\emph{et. al.}~\cite{ywhl-ddcpm-14-journal} propose two
$\Delta+1$ coloring algorithms that do not require the knowledge of
the maximum node degree~$\Delta$.  Their first algorithm runs in
$\O(\Delta \log n + \log^2 n)$ time slots and assumes that nodes are
able to increase their transmission power for the computation. This
prevents conflicts between non-leader nodes by allowing the set of
leaders to directly communicate to other leaders outside the
transmission region and thus coordinating the assignment process.
Their second algorithm does not require this assumption and runs in
$\O(\Delta \log^2 n)$ time slots.

\noindent
Our main contributions are 
\begin{itemize}
\item a very simple algorithm to compute a~$\Delta+1$ algorithm in
 ~$\O(\Delta \log n)$ time slots in the synchronous setting;
\item an abstract method that has the potential of improving the
  runtime of other randomized algorithms in the SINR model by a
 ~$\log n$ factor; and
\item an asynchronous color reduction scheme, which, combined with
  known coloring algorithms computes a~$\Delta+1$ coloring in
 ~$\O(\Delta \log n)$ time slots.
\end{itemize}

The coloring algorithms improve current algorithms in the same setting
(cf. Derbel and Talbi \cite{dt-dncsi-10}) regarding the number of
colors the declared goal of~$\Delta+1$, while the runtime is matched.
Other~$\Delta+1$ coloring algorithms in the SINR model require at
least~$\O(\Delta \log n + \log^2 n)$ time slots (under not comparable
assumptions). The method to improve the runtime by a~$\log n$ factor
carefully combines the uncertainty in randomized algorithms with the
uncertainty in the SINR model to handle them combined in the
analysis. For more details, we refer to the Analysis of
Algorithm~\ref{algo:rand4deltacolor} in \cref{sec:analysis}.

\textbf{Roadmap:} In the next section we state the model along with
required definitions. Communication results in the SINR model required
in the algorithm are stated in \cref{sec:sinr-model-related}.  In
\cref{sec:algorithm} the algorithms for the synchronous setting are
described and analyzed. The color reduction
scheme is generalized to the asynchronous setting in
\cref{sec:async-color-reduct}.

\section{Model and Preliminaries}
\label{sec:prelim-coloring}

Let us formulate some notation before defining the coloring
problem. The \emph{Signal-to-Interference-and-Noise-Ratio} (SINR)
model is used to model whether transmission in a wireless network can
be successfully decoded at the intended receivers or not. We say that
a transmission from a sender to a receiver is \emph{feasible} if it
can be decoded by the receiver.  In the SINR model it depends on the
ratio between the desired signal and the sum of interference from
other nodes plus the background noise whether a certain transmission
is successful. Let each node~$v$ in the network use the same
transmission power~$P$. Then a transmission from~$u$ to~$v$ is
feasible if and only if
\begin{align*}
  \frac{\frac{P}{\dist(u,v)^\alpha}}{\sum_{w \in I}
    \frac{P}{\dist(w,v)^\alpha} + \noise} \geq \beta,
\end{align*}
where~$\alpha, \beta$ are constants depending on the hardware,
$\noise$ reflects the environmental noise,~$\dist(u,v)$ the Euclidean
distance between two nodes~$u$ and~$v$, and~$I \subseteq V$ is the set
of nodes transmitting simultaneously to~$u$.  The \emph{broadcasting
  range}~$r_B$ of a node~$v$ defines the range around~$v$ up to which
$v$'s messages should be received. We denote the set of neighbors of
$v$ by~$N_v := \{w \in V\backslash \{v\} | \dist(v,w) \leq r_B\}$ and
$N_v^+ := N_v \cup \{v\}$. Based on the SINR constraint, the
\emph{transmission range} of~$r_T \leq ( \frac{P}{\beta \noise}
)^{1/\alpha}$ is an upper bound for the broadcasting range (with~$r_B
\leq (1-\epsilon) r_T$ for a constant~$\epsilon > 0$, otherwise no
spacial reuse is possible).  Let the \emph{broadcasting region}~$B_v$
be the disk with range~$r_B$ centered at~$v$.

The communication graph~$G = (V,E)$ is defined as follows. The set of
vertices~$V$ in the graph corresponds to the set of nodes in the
network, while there is an edge~$(u,v) \in E$ if and only if~$u$ and
$v$ are \emph{neighbors} (i.e., they are within each other's
broadcasting range). The \emph{maximum degree} in the network is
$\Delta := \max_{v \in V} |N_v|$.  Note that since~$r_B < r_T$, a
node~$v$ may successfully receive transmissions from nodes that are
not its neighbors in the communication graph, although successful
transmission from those node cannot be guaranteed. As the signal
strength decreases geometrically in the SINR model, we can assume that
messages from outside the broadcasting range are discarded by
considering the signal strength of a received message (wireless
communication hardware usually provides this indication as the
Received-Signal-Strength-Indication (RSSI)
value\cite{bardwell2002converting}). In a more practical setting, one
could also define the communication graph based on the possible
communication between two nodes. For this work we restrict
communication of the nodes to the broadcasting range by assuming nodes
outside the broadcasting range discard the message. To prove
successful communication within the broadcasting range, we need the
concept of a \emph{proximity range}~$r_A > 2 r_B$ around a node~$v$ as
introduced in~\cite{gmw-lbpim-08}. Let~$\Delta^v_A$ be the number of
nodes with distance less than~$r_A$ to~$v$, and
$\Delta^A := \max_{v \in V} \Delta^v_A$. It holds that
$\Delta^A \in \O(\Delta)$. As further technical details of the
proximity range are not required in our analysis, we refer to
\cite{gmw-lbpim-08} for the exact definitions.

In the \emph{synchronous setting}, we assume the nodes to start the
algorithm at the same time. We discuss this assumption briefly in
\cref{sec:discussion}. In the more realistic \emph{asynchronous
  setting} arbitrary wake-up of nodes is allowed and we do not require
synchronized time slots; decent clocks, however, are assumed. In our
analysis, we shall argue using time slots although the nodes to not
have common time slots. As, for example, in \cite{gmw-lbpim-08}
communication can be established in asynchronous networks similar to
synchronous networks using the aloha trick.

We call two nodes~$v,u \in V$ \emph{independent} if they are not
neighbors. A set~$S \subseteq V$ such that the nodes in~$S$ are
pairwise independent is called \emph{independent set}. Obviously,
$S \subseteq V$ is a \emph{maximal independent set} (MIS) if~$S$ is
independent and there is no~$v \in V\backslash S$ with~$S \cup \{v\}$
independent.  Denote the set of integers~$\{0,\dots,i\}$ by~$[i]$. Let
us now define the coloring problem.  Given a set of nodes~$V$ so that
each node~$v \in V$ has a color ~$c_v$, and let~$d$ be an
integer. Then~$V$ has a \emph{valid~$d+1$ coloring}, if for each
node~$v$ holds~$\forall v \in N_v: c_v \not = c_w$ and~$c_v \in [d]$.
Observe that in a valid coloring each color in the network forms a
independent set.

In the following we introduce the notation required for the analysis
of Algorithm~\ref{algo:randcolor}, which consists of several phases
and is described in the next section. We say that two nodes~$v,w$ have
a \emph{conflict} if~$c_v = c_w$ and denote the temporary color of~$v$
in phase~$t$ by~$c^t_v$. The set of nodes that are in a conflict
with~$v$ in phase~$t$ is~$X^t(v) := \{w \in N_v | c^t_v = c^t_w\}$. We
call~$X^t(v)$ the \emph{conflict set} of~$v$ in phase~$t$.  Let us now
define some events. The event that~$v$ is in a conflict in phase~$t$
is~$\etc(v) := \exists w \in X^t(v)$. Note that it does not matter
whether~$v$ knows of the conflict or not.  The event that a
transmission from~$v$ to all neighbors~$N_v$ of~$v$ in phase~$t$ is
\emph{successful} is~$\ets(v)$. A transmission from~$v$ to its
neighbors in phase~$t$ is not successful or \emph{fails} if at least
one neighbor was unable to receive the message. The corresponding
event is~$\E^t_\text{fail}(v)$. We replace~$\E$ by~$\P$ to denote the
probability of an event, e.g.~$\pts(v)$ instead of~$\ets(v)$.  Note
that although the events~$\ets(v)$, and~$\etf(v)$ may not be
independent of events happening at other nodes, our bounds on the
corresponding probabilities~$\pts(v)$ and~$\ptf(v)$ are independent
from the node~$v$ and possible events at other nodes.  Also, our
bounds~$\pts(v)$ and~$\ptf(v)$ on these events include the event
that~$v$ reaches some but not all of its neighbors,
as~$\ptf(v) \leq 1 - \pts(v) \leq 1/12$
and~$11/12 \leq \pts(v) \leq 1$ (see
Section~\ref{sec:sinr-model-related} for a proof of the bounds).

The nodes use two different transmission probabilities in order to
adapt to the requirements of the corresponding
algorithms. Probability~$p_1 := \frac{1}{2\Delta^A}$ is used in
Algorithm~\ref{algo:rand4deltacolor}, while
Algorithm~\ref{algo:colorreduction} uses~$p_2 := \frac{1}{180}$.
\cref{algo:async-color-reduction} uses both transmission
probabilities.  Let~$c$ be an arbitrary constant with~$c > 1$.
Throughout the paper, we use the following constants:
$\kappa_\ell := c \lambda \ln n / p_\ell$ for~$\ell=1,2$,
$\kappa_0 := \lambda \ln 12 / p_1$,
and~$\lambda := \left( P^{A_v}_\text{none} \cdot P^v_\text{SINR}
\right)^{-1}$,
where~$P^{A_v}_\text{none} = (1/4)^{\O(1)}$
and~$P^v_\text{SINR} = 1/2$ to simplify arguments in the SINR model
(for more details and a derivation of the probabilities, we refer to
Appendix~\ref{sec:succ-transm-with}).  Note that
$\kappa_0 \in \O(\Delta)$, $\kappa_1 \in \O(\Delta \log n)$,
and~$\kappa_2 \in \O(\log n)$.  We shall use the following
mathematical fact in our analysis. It can be found, for example, in
the mathematical background section of \cite{mr-ra-10}.
\begin{fact}
  \label{fact:math-fact}
  For all~$t,n \in \mathbb{R}$, such that~$n \geq 1$ and~$|t| \leq n$
  it holds that~$\left(1+\frac{t}{n}\right)^n \leq e^{t}$.
\end{fact}

\section{Adaptable Transmission Probabilities in the SINR Model}
\label{sec:sinr-model-related}

We shall show in this section that local broadcasting with constant
success probability in time reversely proportional to the transmission
probability can be achieved. This is extends known results regarding
local broadcasting, which guarantee local broadcasting with
\emph{high} probability for a fixed number of time slots.  As in these
results we require the nodes transmission probability to conform with
the requirement that the sum of transmission probabilities from each
broadcasting region is at most 1, which is stated in
Lemma~\ref{lem:sum-tx-prob}. This local property is used in the
following theorem to ensure that both the interference by few local
transmissions as well as the summed interference by all globally
simultaneous transmissions can be handled.
\begin{restatable}{lemma}{lemmasumtxbounded}
  \label{lem:sum-tx-prob}
  Let all nodes in the network execute Algorithm~\ref{algo:randcolor}
  (or \cref{algo:async-color-reduction}), and let~$v$ be an arbitrary
  node.  Then the sum of transmission probabilities from within~$v$'s
  broadcasting range is at most 1.
\end{restatable}
The proof is a simple summation over the number of nodes in the
broadcasting region and the transmission probabilities used for nodes
in the algorithms. Due to space restrictions it is deferred to
Appendix~\ref{sec:omitted-proofs}.  Let us now state the main theorem
of this section.

\begin{restatable}{theorem}{sinrthm}
  \label{thm:sinr}
  Given a network of nodes. Let~$v$ be an arbitrary node transmitting
  with probability~$p>0$ in each time slot, and let the sum of
  transmission probabilities from within each broadcasting region be
  at most 1. A transmission from~$v$ is successful within
 ~$\frac{\lambda \ln 12}{p}$ time slots with probability at
  least~$\frac{11}{12}$.
\end{restatable}

The proof of the theorem is along the lines of the proof that local
broadcasting can be achieved in~$\O(\Delta \log n)$ time slots
\cite{gmw-lbpim-08}. Thus, we shall only give a sketch of the proof
here. For the sake of completeness, an adaptation of the full proof to our
theorem can be found in Appendix~\ref{sec:succ-transm-with}.

\begin{proof}
  The proof is based on two major observations, which are both implied
  by the fact that the sum of transmission probabilities from within
  each broadcasting range is bounded by 1.
\begin{enumerate}
\item The probability~$P^{A_v}_\text{none}$ that~$v$ is the only node
  in the proximity region that transmits a signal in the current time
  slot is constant.
\item The probability~$P^v_\text{SINR}$ that the SINR constraint
  holds for a given transmission of~$v$ is constant. 
\end{enumerate}
By combining these probabilities with the transmission probability it
follows that the probability that~$v$ successfully transmits a message
to all neighbors in a given time slot is at least
$p \cdot P^{A_v}_\text{none} \cdot P^v_\text{SINR}$. As we set
$\lambda := \left( P^{A_v}_\text{none} \cdot P^v_\text{SINR}
\right)^{-1}$,
the probability for a successful transmission to all neighbors
after~$\kappa_0 = \frac{\lambda \ln 12}{p}$ time slots is at least
\begin{align*}
  1- \left( 1 - \frac{p}{\lambda} \right)^{\frac{\lambda \ln 12}{p} }
  \geq 1- e^{-\ln 12} \geq 1- \frac{1}{12} \geq \frac{11}{12}, 
\end{align*}
which proves the theorem.
\end{proof}

A combination of Lemma~\ref{lem:sum-tx-prob} and
Theorem~\ref{thm:sinr} implies the following two results. The first
states that in each phase of Algorithm~\ref{algo:rand4deltacolor}, the
currently selected color is transmitted to all neighbors with constant
probability, which enables us to analyze Algorithm~\ref{algo:rand4deltacolor} and
correctly account for the uncertainty in the message transmission in
this algorithm in \cref{sec:analysis}.  
\begin{corollary}
  \label{cor:randcolor-const-tx-succ}
  Let~$v$ be a node in the network executing
  Algorithm~\ref{algo:rand4deltacolor}. Then~$v$ successfully
  transmits its color to its neighbors with
  probability~$\pts(v) \geq \frac{11}{12}$ in each phase of
  length~$\kappa_0 \in \O(\Delta)$.
\end{corollary}
The second result applies to the message transmission in the color
reduction scheme of Algorithms~\ref{algo:colorreduction}
and~\ref{algo:async-color-reduction}. In order to make the lemma
applicable to both algorithms, we state the result in a general
way. It includes the standard local broadcasting bounds known from
\cite{gmw-lbpim-08}, but also enables the algorithms to use a more
coordinated (and faster) medium access based on the coloring. Local
broadcasting achieves successful message transmission from each node
in the network in~$\O(\Delta \log n)$ time slots, while the more
coordinated approach used in the color reduction scheme achieves
successful message transmission from a constant number of nodes in
each broadcasting region to their neighbors in~$\O(\log n)$ time
slots.

\begin{lemma}
  \label{lem:sinr-lb-and-cr-tx}
  Given a network of nodes. Let~$v$ be a node transmitting
  with probability~$p_\ell$ in each time slot according to
  Algorithm~\ref{algo:colorreduction} or~\ref{algo:async-color-reduction}. Then a
  transmission from~$v$ is successful within~$\kappa_\ell$ time slots
  w.h.p. for~$l \in \{1,2\}$.
\end{lemma}

\begin{proof}
  Note that most nodes transmit with probability~$p_1$ in both
  algorithms and only few, i.e., a constant number of nodes in each
  broadcasting range, are allowed to transmit with probability~$p_2$
  simultaneously. Having all nodes transmitting with~$p_1$ for
 ~$\kappa_1 \in \O(\Delta \log n)$ time slots is essentially the
  well-known result on local broadcasting by Goussevskaia
  \emph{et. al.} \cite{gmw-lbpim-08}.  The transmission
  probability~$p_2$ is constant and selected such that the constant
  number of nodes required by the respective algorithms may use~$p_2$
  without violating \cref{lem:sum-tx-prob}.

  Similar to the argumentation in \cref{thm:sinr}, attempting to
  transmit with probability~$p_\ell$ for~$\kappa_\ell$ time slots
  ($\ell=1,2$) yields a failure probability of at most
\begin{align*}
  1- \left( 1 - \frac{p_\ell}{\lambda} \right)^{\kappa_\ell}
  \geq 1- e^{c\ln n} \geq 1- \frac{1}{n^c},
\end{align*}
which completes the proof.
\end{proof}

This implies that the communication used in our algorithms is
successful with the corresponding probabilities. We shall prove the
correctness of our algorithms in the following.

\section{Synchronous Setting}
\label{sec:algorithm}

The algorithm we propose is at its heart a very simple and well-known
randomized coloring algorithm. The underlying approach dates back to
an algorithm to compute maximal independent sets by Luby
\cite{l-spamis-86}, and is covered for example in
\cite[Chapter~10]{be-dcg-13}.  Essentially, this kind of algorithms
draw a random color whenever two neighboring nodes have the same color
(i.e., there is a conflict between them).  However, in contrast to
previous algorithms of this kind, we do not assume that successful
communication is guaranteed by lower layers. Instead we allow the
uncertainty in the randomized algorithm to unite with the uncertainty
in the communication in the SINR model, which is jointly handled in
the analysis.  Thereby we can reduce the number of time slots required
for each phase by a~$\log n$ factor (from~$\O(\Delta \log n)$ for the
trivial analysis to~$\O(\Delta)$), making this simple approach
viable in the SINR model.

\LinesNumberedHidden
\SetKw{KwAlgA}{\refstepcounter{subalgo}Algorithm \thesubalgo:\label{algo:rand4deltacolor}}
\SetKw{KwAlgB}{\refstepcounter{subalgo}Algorithm \thesubalgo:\label{algo:colorreduction}} 
\begin{algorithm}
\DontPrintSemicolon
\hspace{-0.25cm}\KwAlgA{{\sc Rand4DeltaColoring} \textnormal{for $v$}} \; 
\blockline
\setcounter{AlgoLine}{0}
\ShowLn $F_v \gets [4\Delta]$, $c_v^{-1} \gets F_v.$random() \;
\ShowLn \For(\tcp*[f]{each one phase}){$t\gets 0; t
  \leq 6 (c+3) \ln n; t \gets t+1$} {  
\ShowLn \lIf(\tcp*[f]{if conflict, new color}){$c^{t-1}_v \not \in F_v
$} {$c^t_v \gets F_v$.rand()}
\ShowLn  \lElse(\tcp*[f]{otherwise, keep it}){$c^t_v \gets c^{t-1}_v$}
\ShowLn $F_v \gets [4\Delta]$ \;
\ShowLn  Transmit $c^t_v$ with probability $p_a$ for $\kappa_0$ time slots\;
\ShowLn  \lForEach {received color $c^t_w$ from neighbor $w \in N_v$} {
 $F_v \gets F_v \backslash \{c^t_w\}$
  }
}
\blockline  
\hspace{-0.25cm}\KwAlgB{{\sc ColorReduction} \textnormal{for node $v$}} \;
\blockline 
\setcounter{AlgoLine}{7} 
\ShowLn $F_v \gets [\Delta]$ \;
\ShowLn \ForEach {color $c_i = i \in [4\Delta]$} {
\ShowLn \If(\tcp*[f]{wait}){$c^t_v \not = i$} {
\ShowLn listen for $\kappa_2$ time slots\;
\ShowLn  \lForEach {received color $c_w$ from neighbor $w \in N_v$} {
      $F_v \gets F_v \backslash {c_w}$
    }
  }
\ShowLn  \Else(\tcp*[f]{otherwise, select and transmit final color}){
\ShowLn    $c_v \gets F_v$.random()\;
\ShowLn    Transmit $c_v$ with probability $p_2$ for $\kappa_2$ time slots\;
  }
}
\caption{{\sc RandColoring} for node $v$}
\label{algo:randcolor}
\end{algorithm}

\begin{figure*}[tb]
  \centering
    \includegraphics[width=0.95\textwidth]{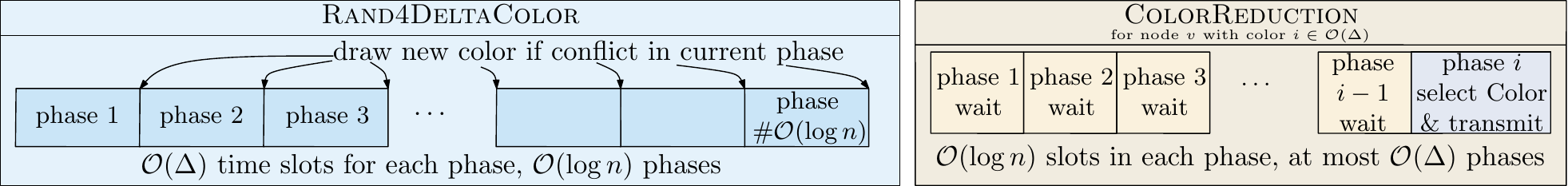}
    \caption{Illustrations of the synchronous Algorithms. Left:
      Algorithm~\ref{algo:rand4deltacolor}. Right: Algorithm~\ref{algo:colorreduction}.}
    \label{fig:sync-algorithm-illustration}
\end{figure*}

Our first algorithm (\RandCDeltaColoring,
Algorithm~\ref{algo:rand4deltacolor} and
\cref{fig:sync-algorithm-illustration}) is a simple, phase-based
coloring algorithm.  In each phase the node checks whether it knows of
a conflict with one of its neighbors. If so, it randomly draws a new
color from the set $F_v$ of colors not taken in the previous
phase. Finally, the phase is concluded by transmitting the current
color. This part computes a valid coloring with~$4\Delta$ colors
in~$\O(\log n)$ phases, while each phase takes~$\O(\Delta)$ time
slots.  Our main contribution regarding this algorithm lies in the
analysis, as phases of length~$\O(\Delta)$ are not sufficient to
guarantee successful message transmission. This allows us to save a
$\log n$ factor, at the cost of handling uncertain message
transmission in the analysis of the algorithm.

Assume the nodes have agreed on a valid coloring with~$d$ colors.  The
second algorithm (\ColorReduction, Algorithm~\ref{algo:colorreduction}
and \cref{fig:sync-algorithm-illustration}) improves the coloring by
reducing the number of colors to~$\Delta+1$.  The general scheme is
that the nodes in the networks gradually replace their color by a
color from $[\Delta]$. In each phase, the nodes of exactly one color
select a new color from~$[\Delta]$ and broadcast their selection to
their neighbors. As each color forms an independent set, those nodes
are free to select a color that is not yet taken in their neighborhood
and communicate the selected color to the neighbors without
introducing a conflict.  In order to make the algorithm efficient in
the SINR model, we utilize the given coloring to improve the
communication efficiency by building a tentative schedule for the
transmissions of each node. As nodes of exactly one color attempt to
transmit in each phase, the nodes of this color can successfully
transmit their (new) color selection to neighboring nodes with high
probability in only~$\O(\log n)$ time slots instead
of~$\O(\Delta \log n)$. As each node that selects a final color
from~$[\Delta]$ knows of its neighbors colors, it can select and
communicate its color without introducing a conflict. Overall this
algorithm requires~$d$ phases to reduce a~$d$ coloring to a~$\Delta+1$
coloring. As each phase requires $\O(\log n)$ time slots this results
in~$\O(d \log n)$ time slots.

Our algorithms are combined to compute a~$\Delta+1$ coloring from
scratch.  Directly after finishing the execution of
Algorithm~\ref{algo:rand4deltacolor}, the nodes begin executing
Algorithm~\ref{algo:colorreduction} in order to compute a $\Delta+1$
coloring from scratch in~$\O(\Delta \log n)$ time slots.  Both parts
of the algorithm are well-known in the $\mathcal{LOCAL}$ model,
however, simply executing the algorithms in the SINR model leads to a
runtime~$\O(\Delta \log^2 n + \Delta^2 \log n)$ time slots.  Due to
the careful combination of the algorithms with specifics of the SINR
model, along with handling the uncertainty of successful message
transmission both in the algorithm and the analysis, we are able to
execute both parts of the algorithm in $\O(\Delta \log n)$ time
slots. Thus, Algorithm~\ref{algo:randcolor} solves the~$\Delta+1$
coloring problem in~$\O(\Delta \log n)$ time slots, which matches the
runtime of one execution of local broadcasting in the SINR model.

\subsection{Analysis of \RandCDeltaColoring}
\label{sec:analysis}

Despite the fact that the underlying coloring algorithm is well-known,
our analysis is new and quite involved.  The main reason for this is
the uncertainty in whether a message is successfully delivered in one
phase of Algorithm~\ref{algo:rand4deltacolor}. In contrast to
guaranteed message delivery based for example on local broadcasting,
message delivery with constant probability can be achieved a
logarithmic factor faster, see Section~\ref{sec:sinr-model-related}.
However, this reduction in runtime comes at a cost: While in the
guaranteed message delivery setting, a node can finalize its color
once a phase without a conflict happened, this is not possible in our
setting, as we cannot guarantee the validity of the colors even if a
node did not receive a message implying a conflict in one phase due
to message transmission with only constant probability.  Nevertheless, we can show
that after~$\O(\log n)$ phases of transmitting the selected color and
resolving eventual conflicts, the coloring is valid in the entire
network w.h.p.

Our formal analysis shall be structured as follows. We show in this
section that the probability for a conflict indeed reduces by a
constant factor in each phase. This is the foundation for the result
that the first part of our coloring algorithm computes a valid
$4\Delta$ coloring in~$\O(\Delta \log n)$ time slots w.h.p., which is
the main result in this section.  In \cref{sec:color-reduction} we
show that Algorithm~\ref{algo:colorreduction} transforms a valid
$4\Delta$ coloring into a valid~$\Delta+1$ coloring with high
probability and conclude that the total runtime of
\cref{algo:randcolor} is~$\O(\Delta\log n)$ time slots.

In order to prove correctness of Algorithm~\ref{algo:rand4deltacolor}
(\RandCDeltaColoring) we shall first bound the probability of a
conflict propagating from one phase of the algorithm to the next.
This enables us to show that each node that participates in the
algorithm selects a valid color.  Based on the probability that a
node~$v$ has a conflict in phase~$t$, there are only two cases that
may lead to a conflict at~$v$ in phase~$t+1$:
\begin{enumerate}
\item $v$ had a conflict in phase~$t$, and it did not get resolved
  (either due to being unaware of the conflict or since the new color
  implies a conflict as well).
\item a neighbor of~$v$ had a conflict in phase~$t$ and introduced
  the conflict by randomly selecting~$v$'s color.
\end{enumerate}
We shall show that the probability for both cases is constant (see
Lemma~\ref{lem:phase-conflict-prob}). Thus, after~$\O(\log n)$
phases it holds with high probability that a valid color has
been found. Let us now state the main result of this section.

\begin{theorem}
  \label{thm:colors-valid-whp}
  Let all nodes start executing Algorithm~\ref{algo:randcolor}
  simultaneously. After the execution, all nodes have a valid color
  with probability~$1-1/n^{c}$.
\end{theorem}

Before proving this result let us formulate and prove the following
lemma, which lays the foundations to prove
Theorem~\ref{thm:colors-valid-whp}.  In this lemma we bound the
probability of a conflict in phase~$t+1$ based on the probability of a
conflict in phase~$t$.

\begin{lemma}
  \label{lem:phase-conflict-prob}
  Let~$v$ be an arbitrary node and~$\ptc(v)$
  the probability of a conflict at~$v$ in phase~$t$. Then the
  probability of a conflict at~$v$ in phase~$t+1$ is at most
  \[\pc^{t+1}(v) \leq \frac{5}{6} \cdot \max_{w \in
  N_v} \ptc(w).\]
\end{lemma}

\begin{proof}
  We shall prove the lemma by considering the two cases that may lead
  to a conflict at node~$v$ in phase~$t+1$. The \textbf{first} case is
  that~$v$ has a conflict with at least one of its neighbors.
  Depending on which transmissions are successful there are 3
  subcases. Note that $\rightarrow$ denotes $\ets(v)$, while
  $\leftarrow$ denotes $\exists w \in X^t(v): \ets(w)$---with
  negations accordingly\footnote{A partial success of transmission is
    often sufficient to trigger dealing with a conflict.  We do not
    consider this in our notation, however, as we evaluate $\pts(v) $
    to be at most 1 for all~$v$ and since $\Pr(\text{transmission from
      $v$ to $u$ fails}) \leq \ptf(v) \leq 1/12$, our analysis covers
    this case.}.

  \begin{enumerate}[leftmargin=0pt,itemindent=48pt]
  \item[\textbf{(a)} $\not\rightarrow, \not\leftarrow$:] It is not
    guaranteed that any of the conflict partners know of the conflict,
    as the transmissions from~$v$ and the nodes in the conflict set
    $X^t(v) \neq \emptyset$ failed at least partially. There is at
    least one neighbor~$u \in X^t(v)$ that failed to transmit its
    color successfully to $v$, which happens with probability
    $\ptf(u)$.  Combined with $v$'s failure to transmit its color
    successfully, case~1(a) happens with probability at most
    $\ptc(v) (\ptf(v) \Pr(\not\leftarrow))$
    $\leq \ptc(v) \ptf(v) \ptf(u) \leq \ptc(v) (1/12)^2$. If any
    conflict partner knows of the conflict, the conflict would be
    resolved with a certain probability (as in the following
    cases). However, as this is not guaranteed, we account for the
    worst case, which is that the conflict is not resolved and
    propagates to the next phase.  Note that since this case happens
    only with a small probability, we have shown that the total
    probability of case~{(a)} and conflict at~$v$ in phase~$t+1$ is
    small.

  \item[\textbf{(b)} $\rightarrow, \not\leftarrow$:] All nodes in
    $X^t(v)$ failed to transmit successfully, but~$v$ transmitted
    successfully to all neighbors. Thus, all nodes in $X^t(v)$ know of
    the conflict, while~$v$ might be unaware of it. This case happens
    with probability at
    most~$\ptc(v) \cdot (\pts(v) \cdot \Pr(\not \leftarrow))$. The
    probability that a node~$w \in X^t(v)$ selects~$v$'s color in
    phase~$t+1$ is at most~$\sum_{w \in X^t(v)} 1/|F_w|$ (even if~$v$
    knows of a conflict and itself selects a new color). This results
    in an overall probability of at most
    \begin{align*}
      & \ptc(v) \cdot (\pts(v) \cdot
      \Pr(\not \leftarrow)) \cdot \sum_{w \in X^t(v)}
      \frac{1}{|F_w|} \\
      & \leq \ptc(v) \left( \prod_{w \in X^t(v)}
        \ptf(w) \right) \cdot \sum_{w \in X^t(v)}
      \frac{1}{|F_w|} \\
      & \overset{x := |X^t(v)|}{\leq} \ptc(v) \left( \ptf \right)^x
      \cdot \frac{x}{3\Delta} \\
      & \leq \frac{1}{3\Delta} \ptc \cdot x \left(
        \frac{1}{12} \right)^x \leq \frac{1}{24} \ptc
    \end{align*}
    where the first inequality holds since the event~$\not\leftarrow$
    is equivalent to~$\forall w \in X^t(v): \etf(w)$ and~$\pts(v) \leq
    1$. The second inequality holds since~$|F_w| \geq 3\Delta$ as~$w$
    and~$v$ are uncolored and by setting~$x = |X^t(v)|$. The last
    inequality holds since~$x (1/12)^x \leq 1/12$ for all~$x \in
    \{1,\dots,\Delta\}$, and~$\Delta \geq 1$.
    \addtolength{\itemindent}{-15pt}
  \item[\textbf{(c)} $\leftarrow$:] It holds that~$v$ knows of the
    conflict. Whether the neighbors know of it or not is not
    guaranteed.  This case happens with probability at most
    $\ptc(v) \cdot (\Pr(\leftarrow))$. The probability that at least
    one neighbor of~$v$ has or selects the same color as~$v$ is at
    most
    \begin{align*}
      \sum_{w \in N_v} \frac{1}{|F_v|} \leq |N_v|
      \frac{1}{3\Delta} \leq \frac{1}{3}.
    \end{align*}

  \end{enumerate}

  Overall this results in a probability for a conflict at~$v$ in
  phase~$t+1$ of at most
  \begin{gather*}
    \ptc(v) \cdot \left( \frac{1}{144} + \frac{1}{24} + \frac{1}{3}
      \Pr(\leftarrow) \right) \\
    \leq \; \ptc(v) \cdot \left( \frac{1}{144} + 
      \frac{1}{24} + \frac{1}{3} \right) < \; \left( \frac{1}{2}
    \right) \cdot \ptc,
  \end{gather*}
  where the first inequality follows from $\Pr(\leftarrow) \leq 1$.

  The \textbf{second} case is that there was no conflict at~$v$ in
  phase~$t$, but a neighbor~$w$ of~$v$ selected~$v$'s color due to a
  conflict at~$w$. The probability of this event is at most
  \begin{gather*}
    \sum_{w \in N_v} \underbrace{\Pr(c^{t+1}_v =
    c^{t+1}_w)}_{\text{$v$'s neighbor $w$ selects $v$'s color}} \sum_{u \in N_w}
  \underbrace{ \Pr(c^t_u = c^t_w) 
    }_{
      \substack{\text{$u \in N(w)$ told $w$}\\
      \text{about their conflict}}
} \\
    \leq \sum_{w \in N_v} 
    \Pr(c^{t+1}_v = c^{t+1}_w) \ptc(w) \\
    \leq \; \sum_{w \in N_v}  
    \frac{1}{|F_w|} \ptc(w) \;\;
    \leq \; \left( \frac{1}{3} \right) \max_{w \in N_v} \ptc(w)
  \end{gather*}
  The last inequality holds since
  $\sum_{w \in N_v} \frac{1}{|F_w|} \leq \sum_{w \in N_v}
  \frac{1}{3\Delta} \leq \frac{1}{3}$.
  Combining all events that could lead to a conflict at~$v$ in
  phase~$t+1$ it holds that the probability of the union of the events
  is at most
  \begin{align*}
    \P^{t+1}_\text{confl}(v) &\leq \left( \frac{1}{2}\right) \ptc(v) +
    \left( \frac{1}{3} \right) \max_{w \in N_v} \ptc(w) \\
    & \leq \frac{5}{6} \cdot \max_{w \in N_v^+} \ptc(w),
  \end{align*}
  which concludes the proof.
\end{proof}

Note that the second case could be avoided if message delivery in each
phase would be guaranteed, as a node~$v$ that does not have a conflict
in phase~$t$, would simply finalize its current color and communicate
this. Thus, $v$ could not be forced into a conflict anymore. We shall
now show that a set of nodes executing
Algorithm~\ref{algo:rand4deltacolor} computes a valid coloring.

\begin{proposition}
  \label{prop:randcolor-valid-whp}
  Let all nodes in the network execute
  Algorithm~\ref{algo:rand4deltacolor} simultaneously and let~$v$ be
  an arbitrary node. Then the probability that~$v$ has a conflict with
  another node~$u$ after the execution is at most $\frac{1}{n^{c+3}}$.
\end{proposition}

\begin{proof}
  Let us consider the probability of a conflict at an arbitrary
  node~$v \in V$ in phase~$t = 6(c+3)\ln n$. It holds that
  \begin{align*}
    & \ptc(v) \leq
      \left(\frac{5}{6}\right) \max_{w \in N_v} \pc^{t-1}(w)  
      \leq \left(\frac{5}{6}\right) \max_{w \in V} \pc^{t-1}(w) \\ 
    &\leq
      \left(\frac{5}{6}\right)^t \max_{w \in V} \pc^{0}(w) 
    \leq \left(1 - \frac{1}{6}\right)^{6(c+3)\ln n} \leq
      \frac{1}{n^{c+3}},
  \end{align*}
  where the first inequality is due to
  Lemma~\ref{lem:phase-conflict-prob}. The third inequality holds
  since all nodes are in the same phase due to the synchronous start
  of the algorithm. Note that the upper bound on the probability that
  a conflict propagates holds for all nodes. The fourth inequality
  holds as $\pc^{0}(v)$ is obviously bounded from above by~$1$ for all
  nodes~$v$. The last inequality holds due to
  Fact~\ref{fact:math-fact}.
\end{proof}

It is guaranteed by Proposition~\ref{prop:randcolor-valid-whp} that
the nodes that finished Algorithm~\ref{algo:rand4deltacolor} have
valid colors from~$[4\Delta]$. It remains to show that
Algorithm~\ref{algo:colorreduction} reduces the colors from $4\Delta$
to $\Delta + 1$.

\subsection{Synchronous Color Reduction}
\label{sec:color-reduction}

We shall only sketch the proof in the following, as the general color
reduction scheme is already known for the $\mathcal{LOCAL}$ model
(cf. for example \cite{be-dcg-13}).

\begin{lemma}
  \label{lem:color-reduction-correct}
  Given a network such that each node has a valid
  color~$c_v \in [4\Delta]$ with probability at
  least~$1-\frac{1}{n^{c+3}}$. Then a synchronous execution of
  Algorithm~\ref{algo:colorreduction} computes a valid $\Delta+1$
  coloring at~$v$ with probability at least~$\frac{1}{n^{c+1}}$.
\end{lemma}

\begin{proof}
  Given a node~$v$ with valid color with
  probability~$1-\frac{1}{n^{c+3}}$. As the color reduction scheme is
  a deterministic algorithm (apart from the communication), only two
  possibilities for an invalid color at~$v$ after the execution exist.
  The first is that $v$'s color was not valid, and $v$ selected the
  same color as one of its neighbors, as both selected their final
  color in the same phase of Algorithm~\ref{algo:colorreduction}. By
  applying a union bound, this happens with probability at most
  $\frac{\Delta}{n^{c+3}}$. The second is that~$v$ did not receive the
  color of one of its neighbors, or one of~$v$'s neighbors did not
  receive~$v$'s color.  It follows from \cref{lem:sinr-lb-and-cr-tx}
  that communication in Algorithm~\ref{algo:colorreduction} is
  successful with probability~$1-\frac{1}{n^{c+3}}$.  Thus, applying
  another union bound, it follows that the overall probability for an
  invalid color at~$v$ is at
  most~$\frac{2\Delta+1}{n^{c+3}} \leq \frac{1}{n^{c+1}}$.
\end{proof}

We can now prove the correctness and runtime of the algorithm.

\begin{proof}[Proof of Theorem~\ref{thm:colors-valid-whp}]
  By Proposition~\ref{prop:randcolor-valid-whp}, after executing
  Algorithm~\ref{algo:rand4deltacolor} (\RandCDeltaColoring), a node
  has a conflict with probability at most~$\frac{1}{n^{c+3}}$. Using
  Lemma~\ref{lem:color-reduction-correct}, this $4\Delta$ coloring is
  refined to a valid~$\Delta+1$ coloring by
  Algorithm~\ref{algo:colorreduction} (\ColorReduction) with
  probability at least~$1 - \frac{1}{n^{c+1}}$ for each node. A
  union bound over all nodes implies the validity of the coloring with
  high probability.

  The total runtime of the algorithm comprises the runtimes of
  Algorithm~\ref{algo:rand4deltacolor} and
  Algorithm~\ref{algo:colorreduction}.
  Algorithm~\ref{algo:rand4deltacolor} consists of
  $6(c+3) \ln n = \O(\log n)$ phases, and each phase takes
  $\kappa_0 = \O(\Delta)$ time slots according to
  Corollary~\ref{cor:randcolor-const-tx-succ}. As local computations
  do not require a separate time slot, this results in a runtime
  of~$\O(\Delta \log n)$.  In Algorithm~\ref{algo:colorreduction}, for
  each color in~$[4\Delta]$, a node either listens for incoming
  messages or transmits its selected final color
  for~$\kappa_2 \in \O(\log n)$ time slots. This results in an
  additional $\O(\Delta \log n)$ time slots. Therefore, the total
  runtime of the algorithm is~$\O(\Delta \log n)$.
\end{proof}


\newcommand{\slan}{\ensuremath{k}} 
\newcommand{\misruntime}{\ensuremath{(\slan)\kappa_2}}
\newcommand{\slai}{\ensuremath{2 \slan^2 \kappa_2}} 

\section{Asynchronous Color Reduction}
\label{sec:async-color-reduct}

Let us now consider the asynchronous setting, which allows nodes to
wake up at arbitrary times, and does not assume synchronized time
slots apart from in the analysis (cf. \cref{sec:prelim-coloring}).
The algorithm described in this section is the same classic color
reduction algorithm as Algorithm~\ref{algo:colorreduction}, however,
generalized to the asynchronous setting. We assume a valid node
coloring with $d$ colors to be given and reduce the number of colors
to $\Delta+1$ in $\O(d \log n)$ time slots. Let us for the rest of
this section assume that we are given an $\O(\Delta)$ coloring.

\begin{figure*}[b!]
  \centering
    \includegraphics[width=0.95\textwidth]{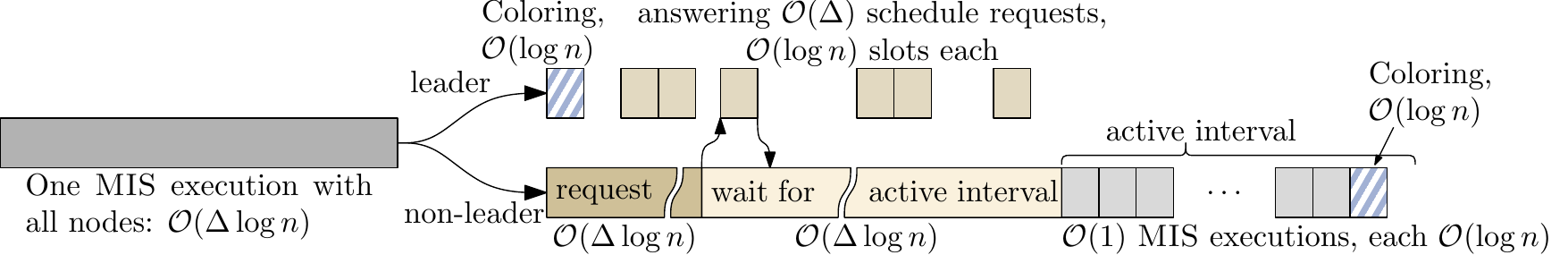}
    \caption{Runtime. Overall $\O(\Delta \log n)$, given a
      $\O(\Delta)$ coloring.}
    \label{fig:async-cr-runtime}
\end{figure*}

Recall that for the synchronous algorithm, a schedule is created based
on the coloring, each node is activated based on its color, and the
nodes can simply select a color from~$[\Delta]$ as within each
broadcasting region at most five nodes are active at each time.  In the
asynchronous setting, however, the nodes cannot simply decide on a
common schedule without synchronization. Solving the synchronization
problem is in $\Omega(D)$ since all nodes in the network must agree on,
or at least be aware of the common schedule. The algorithm we present
in this section circumvents this problem, essentially, by using two
levels of MIS executions. Our algorithm is illustrated in
\cref{fig:async-cr-runtime}, the corresponding pseudocode can be found
as
\cref{algo:async-color-reduction,algo:mis,algo:Colored-l,algo:level-2}.
We reference the {\sc MIS} (\cref{algo:mis}) executed with
parameter~$\ell=1$ by \emph{first level MIS}, and MIS$(\ell=2)$ by
\emph{second level MIS}.  Before introducing the notation used in the
pseudocode, we shall describe the algorithm in more detail.

The algorithm starts by executing the first level MIS$(\ell=1)$
algorithm that determines a set of independent nodes, which we call
leaders.  Each leader node transitions to \cref{algo:Colored-l},
selects and transmits the color 0 it selected and initializes its
periodic leader schedule. This schedule assigns each color an
\emph{active interval} of length $\O(\log n)$ time slots to allow the
nodes of this color to select their final color from $[\Delta]$.

Each node $v_i$ that is not in the first level MIS selects an
arbitrary leader and requests the relative time until it is $v_i$'s
turn to be \emph{active}.  Upon receipt of its active intervall, the
node waits until the interval starts before executing a second level
MIS algorithm (which does not interfere with first level MIS) for a
constant number of times.  In this second level MIS the algorithm
benefits from fewer active nodes, and hence more efficient
communication to allow each node to achieve successful transmission of
a message to all neighbors in $\O(\log n)$ time slots.  Moreover, we
can speed up the MIS algorithm by the same factor of $\Delta$ to
execute in $\O(\log n)$ time slots, as only a constant number of nodes
compete to be in each second level MIS. It holds for each node that
wins the second level MIS, that there is no other node of the second
level MIS in its broadcasting range.  Thus, the winning node can
select a valid color from $\{1,\dots,\Delta\}$ and transmit its choice
to its neighbors without a conflict. If a node does not succeed to be
in the second level MIS, it simply executes MIS(2) again. As each node
succeeds in such an MIS within its active interval, each node selects
one of the $\Delta+1$ colors.

\begin{algorithm}
\DontPrintSemicolon
\ShowLn $F_v \gets [\Delta] \backslash \{0\}$ \;
\ShowLn  \ForEach (\tcp*[f]{do continuously}){received color$_w$ from
  $w$} {
  $F_v \gets F_v \backslash \{$color$_w\}$
}
\ShowLn {\sc MIS(1)}
\caption{{\sc AsyncColorReduction} for node $v$}
\label{algo:async-color-reduction}
\end{algorithm} 

\begin{algorithm}
  \DontPrintSemicolon
  \ShowLn $P_v = \emptyset$,
  $\text{{\sc Next}} = \begin{cases}
    ${\sc Level2}$ &\mbox{if Level 1 MIS }(\ell=1) \\
    ${\sc MIS($2$)}$ & \mbox{otherwise} \end{cases} $
  \ShowLn \For  (\tcp*[f]{Listen first}) {$\kappa_\ell$ time slots} {
    \ShowLn \textbf{for each} $w \in P_v$ \textbf{do}
    $d_v(w) = d_v(w) + 1$\;
    \ShowLn \textbf{if} $M_A^\ell(w, c_w)$ received \textbf{then}
    $P_v = P_v \cup \{w\}$; $d_v(w) = c_w$\;
    \ShowLn \textbf{if} $M_C^\ell(w)$ received \textbf{then} {\sc Next}(w)\;
    }
    \ShowLn $c_v = \Xi(P_v)$\;
    \ShowLn\While (\tcp*[f]{then compete for MIS}){$true$}{
      \ShowLn $c_v = c_v + 1$\;
      \ShowLn \textbf{if} {$c_v > \kappa_\ell$} \textbf{then} {\sc
        Colored}($\ell$)  \tcp*[r]{success}
      \ShowLn \textbf{for each} $w \in P_v$ \textbf{do}
      $d_v(w) = d_v(w) + 1$\;
      \ShowLn \textbf{if} $M_C^\ell(w)$ received \textbf{then} {\sc
        Next}(w) \tcp*[r]{fail}
      \ShowLn transmit $M_A^\ell(v,c_v)$ with probability $p_\ell$\;
      \ShowLn\If (\tcp*[f]{received competing counter}){$M_A^\ell(w, c_w)$ received } {
        \ShowLn$P_v = P_v \cup \{w\}$; $d_v(w) = c_w$\;
        \ShowLn \textbf{if} $|c_v - c_w| \leq \kappa_\ell$ \textbf{then}
        $c_v = \Xi(P_v)$\;
      }
    }
  \caption{{\sc MIS}$(\ell)$ for node $v$, simplified from the MW-coloring algorithm\cite{dt-dncsi-10,sw-cuwm-09}}
  \label{algo:mis}
\end{algorithm}

\begin{algorithm}
  \DontPrintSemicolon
  \ShowLn \eIf (\tcp*[f]{Level 1 leader}){$\ell = 1$}{
    \ShowLn color$_v \gets 0$, $Q \gets \emptyset$, $c'_v = 0$\;
  \ShowLn announce $M_C^1(v,\text{color}_v)$ with prob. $p_2$ for $\kappa_2$ slots\;
    \ShowLn \While(\tcp*[f]{serve requests}) {protocol is executed} {
      \ShowLn $c'_v \gets c'_v + 1$ \;
      \ShowLn transmit $M_C^1(v,\text{color}_v)$ with probability $p_1$\;
      \ShowLn  \ForEach (\tcp*[f]{do continuously}){received
        request $M_R(w,v,\text{color}_w^\text{tmp})$ from neighbor $w$} {
        \ShowLn $Q$.push$((w,\text{color}_w^\text{tmp}))$
      }
      \ShowLn \If {$Q$ not empty}{
        \ShowLn $(w,\text{color}_w^\text{tmp}) \gets Q.$pop(), $t \gets \tau(\text{color}_w^\text{tmp}, c'_v)$\;
        \ShowLn \For (\tcp*[f]{inc./dec. $c_v'$, $t$}){$\O(\log n)$ slots}{
          \ShowLn transmit $M_C^1(v,w,t)$ with probability $p_2$\;
          }
      }
    }
  } (\tcp*[f]{Level 2 / Non-leader node}){ 
  \ShowLn color$_v \gets F_v$.random() \tcp*[r]{no collisions}
  \ShowLn announce $M_C^2(c,\text{color}_v)$ with prob. $p_2$ for $\kappa_2$ slots\;
    \ShowLn \While (\tcp*[f]{keep color valid}){protocol is executed} {
      \ShowLn transmit color$_v$ with probability $p_1$\;
    }
  }
  \caption{{\sc Colored($\ell$)} for node $v$}
  \label{algo:Colored-l}
\end{algorithm}

\begin{algorithm}
  \DontPrintSemicolon
  \ShowLn \While {true} {
    \ShowLn \eIf {$M_C^1(w,v,t)$ received} {
      \ShowLn \While (\tcp*[f]{wait for active interval}) {$t<0$} {
        \ShowLn $t \gets t+1$ \tcp*[r]{one time slot each}
      }
      \ShowLn \While (\tcp*[f]{active interval}){$t < \slai$} {
        \ShowLn \tcp*[l]{increase $t$ by one in each time slot during {\sc MIS}(2)}
        \ShowLn {\sc MIS}$(2)$\;
        }
      }(\tcp*[f]{transmit request}){ 
        \ShowLn transmit $M_R(v,w, \text{color}_v^\text{tmp})$ with probability $p_1$\;
      }
    }
  \caption{{\sc Level2}($w$) for node $v$ with leader $w$}
  \label{algo:level-2}
\end{algorithm}

\subsection{{\sc MIS}, and Notation for {\sc AsyncColorReduction}}
\label{sec:async-notation}

Let us now describe the notation used in the algorithm in more detail,
along with the MIS algorithm, which is simplified from
\cite{dt-dncsi-10,sw-cuwm-09}. We denote the set of available colors
by $F_v$. Note that throughout the algorithm, each node deletes the
final colors it received from $F_v$.
The MIS algorithm (\cref{algo:mis}) aims at allowing exactly one node
in each neighborhood to succeed to Algorithm {\sc Colored}, to select
a color and annouce its success in the MIS algorithm to its
competitors.  There are minor differences depending on the two
levels~$\ell=1$ and~$\ell=2$, however, the algorithm remains the same.
Thus, we describe the algorithm for general~$\ell$.

\cref{algo:mis} is based on the interplay of counters within a node's
neighborhood.  Each node $v$ has a counter $c_v$. We denote the set of
neighbors competing in the MIS by $P_v$. For each node $w$ in this
set, the counter value is stored (and increased by one in each time
slot) as $d_v(w)$. The MIS algorithm begins with a \emph{listen} phase
of $\kappa_\ell$ time slots, which ensures that each node
participating in the MIS received the counter values of other active
neighbors. As nodes joining later execute the listen phase, they know
the status of their neighborhood before competing for being in the
MIS.  Before competing, the counter of $c_v$ of $v$ is set to
$\Xi(P_v)$. This sets $c_v$ to the maximum value such
that~$c_v \leq 0$
and~$c_v \not \in \{d_v(w) - \kappa_\ell,\dots,d_v(w) + \kappa_\ell\}$
for each~$w \in P_v$. Intuitively, $c_v$ is set to a non-positive
value that is not within an interval of $\kappa_\ell$ of a competing
node's counter value.  This ensures that if $v$ reaches the counter
threshold, there is sufficient time to inform the
competitors~$w \in P_v$ of~$v$'s success, or the other way around if
one of~$v$'s competitors succeeds.  Note that the counter values
$d_v(w)$ are kept up-to-date (unless they are reset) by
increasing them in each time slot.

Let us now consider the messages used throughout the algorithm. To
transmit the current counter value $c_v$ to the neighbors, a message
$M_A^l(v,c_v)$ is used. This message contains the level $\ell$, the
transmitting node $v$ and its counter value~$c_v$.  The message that
indicates that a node succeeded in the MIS of level~$l$ is the $M_C^l$
message. This message has two uses in the algorithm. If it is used to
transmit the success in the MIS, it contains the node and the color it
selected. If it is used (by a first level leader $v$) to answer
requests for the activity interval, it contains $v$, the requesting
node $w$ and the time remaining until $w$'s active interval begins
(see below for a definition of $\tau(\cdot,\cdot)$).  Another message
containing only the transmitting node $v$ and its final color
color$_v$ is used in \cref{algo:Colored-l}.  The request message $M_R$
is transmitted by a node $v$ in \cref{algo:level-2}, if $v$ failed to
win the first level MIS algorithm. It contains the node $v$, a leader
$w$ in $v$'s neighborhood and $v$'s color from the initial
$\O(\Delta)$ coloring.

In \cref{algo:Colored-l}, $v$ is a leader, and color$_v$ denotes the final color from
$[\Delta]$, and $Q$ is a queue used to store nodes $w$ along with
their initial color color$_w^\text{tmp}$ that request an active
interval.  The remaining time is based on $v$'s periodic schedule,
which is defined by its counter value $c'_v$, and $w$'s color. We set
$k=90$ and
$\tau(\text{color}_w, c_v)$ such that $-\tau(\text{color}_w, c_v)$ is
positive and minimal with
$-\tau(\text{color}_w, c_v) \equiv \text{color}_w \cdot \slai - c_v
\mod \Delta \slai$
and $\tau(\text{color}_w, c_v) \geq \kappa_2$.  Intuitively, this sets
$t$ to the start of the next interval corresponding to $w$'s color in
$v$'s schedule, so that the starting time can communicated
w.h.p. before the interval starts. Note that $t$ is decreased
appropriately during the transmission interval.

\subsubsection*{Adapting the MIS Algorithm}
\label{sec:adapt-mis-algo}

We assume in the analysis that the MIS algorithm indeed computes a
maximal independent set. This follows directly from the coloring
algorithm in \cite{dt-dncsi-10,sw-cuwm-09}, from which \cref{algo:mis}
was simplified. For reference, we state it as the following lemma.

\begin{lemma}
  \label{lem:mis-independence-runtime}
  \cref{algo:mis} computes a MIS among participating nodes in
  $\delta_\ell \kappa_2$ time slots, where
  $\delta_\ell = \begin{cases} \Delta & \text{if } \ell = 1 \\ k &
    \text{if } \ell = 2 \end{cases}$ w.h.p.
\end{lemma}

Apart from constant changes, the lemma follows directly from
Theorems~1 and~2 in \cite{dt-dncsi-10} if $\ell = 1$. For $\ell=2$, it
follows from \cref{lem:sinr-lb-and-cr-tx} along with setting~$\Delta$
to a constant in the proofs of both theorems. Note that $\delta_\ell \kappa_2 = \O(\delta_\ell \log n)$.

\subsection{Analysis}
\label{sec:analysis-async-color-reduction}

We have seen that union bounding a w.h.p. result decreases the
exponent of the bound (cf. the proof of
\cref{lem:color-reduction-correct}). A higher exponent ultimately
results in increasing the runtime by only a constant factor. Thus, we
refrain from stating exact exponents in our w.h.p. bounds in the
following analysis to simplify notation.  Let us now state the main
result of this section, which we prove at the end of this section.  We
shall prove the lemmata required for its proof in the following.
\begin{theorem}
  \label{thm:async-correct-runtime}
  Given a valid node coloring with $d\geq\Delta$ colors. Then
  Algorithm~\ref{algo:async-color-reduction} computes a valid
  $\Delta+1$ coloring in $\O(d \log n)$.
\end{theorem}
As the algorithm is essentially a simple color reduction scheme, each
node selects a valid color if the communication can be realized as
claimed. To prove this we show that in the second level indeed only a
constant number of nodes are active in each broadcasting range. This
implies that message transmission from active nodes to all their
neighbors can be achieved in $\O(\log n)$ time slots, and similarly
that the second level MIS can also be executed in $\O(\log n)$ time
slots. Finally, we prove that each non-leader node $v$ succeeds in a
second level MIS, and thus colors itself with a color from
$\{1,\dots,\Delta\}$ within the active interval it is assigned by its
leader.

\begin{lemma}
  \label{lem:async-schedule-const-neighbors}
  In the second level, at most $\slan=90$ nodes are active in each
  broadcasting range.
\end{lemma}

\begin{proof}
  The lemma follows from a geometric argument. Let us consider an
  arbitrary node~$v$. Nodes in the broadcasting range of~$v$ must
  select their leader from the set of (first level) MIS nodes within
  the radius of two broadcasting ranges from~$v$. Geometrically at
  most 18 independent nodes can be in such a disk, thus each node from
  the broadcasting range of~$v$ selects one of these 18 nodes as
  leader. Each leader may be selected as leader by at most 5 nodes of
  the same color. Thus, overall the nodes in~$v$'s broadcasting range
  may follow up to~18 schedules, and for each schedule at most 5 nodes
  are active at the same time, which implies the upper bound of~90
  active nodes in the broadcasting range of~$v$. The proof is
  illustrated in \cref{fig:async-color-mis}.
\end{proof}

\begin{figure}[tb]
  \centering
  \includegraphics[width=1.00\linewidth]{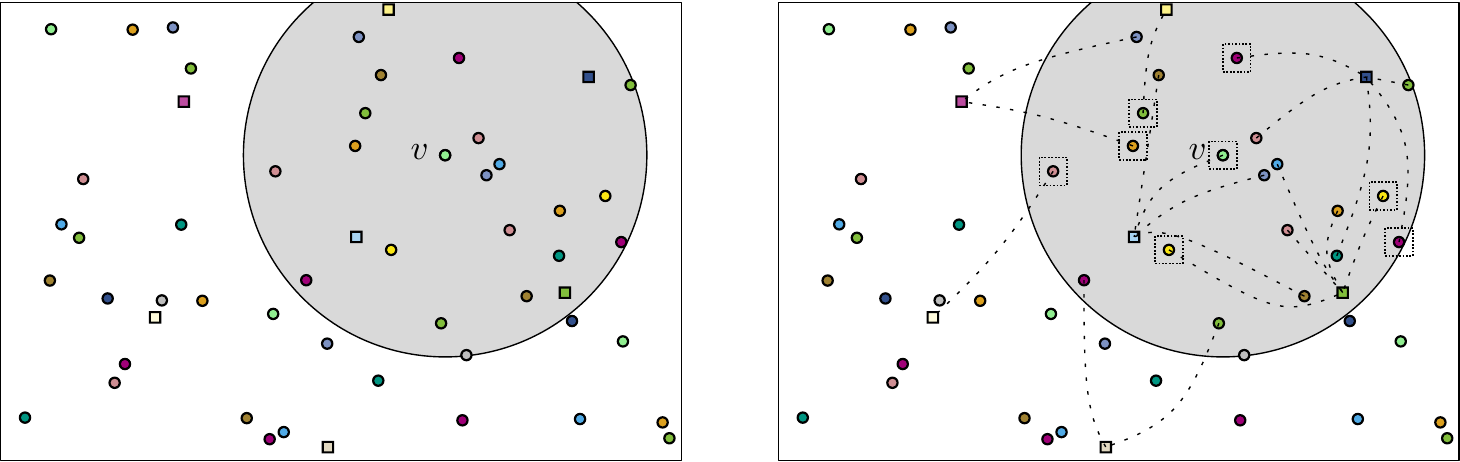}
  \caption{Left: Node $v$ with its broadcasting region in a network
    with valid coloring; Nodes in the first level MIS are
    squares. Right: Nodes in $v$'s broadcasting range are connected to
    their selected leader by a dashed line. Nodes currently active in
    the second level are surrounded by a square. 
  }
  \label{fig:async-color-mis}
\end{figure}

Compared to the classical local broadcasting algorithm, we increase
the transmission probability by a factor of $\Delta$, thus decreasing
the time to successful message transmission by the same factor of
$\Delta$ to $\O(\log n)$. This is possible as instead of~$\Delta$ only
a constant number of nodes try to transmit. Based on this result we
can bound the runtime of our algorithm, starting with
\cref{algo:level-2}.

\begin{lemma}
  \label{lem:async-schedule-level2-wait-time}
  Let $v$ execute \cref{algo:level-2} with leader $w$. Then 
  \begin{inparaenum}[a)]
  \item \label{enum-req-tx} $v$ transmits the request message successfully within $\kappa_1$
    time slots w.h.p.;
  \item \label{enum-receive-ai}  $v$ receives its active interval after at most another
    $\kappa_1$ time slots w.h.p.; and
  \item \label{enum-ai-waittime}$t \leq \Delta \slai \in \O(\Delta \log n)$.
  \end{inparaenum}
\end{lemma}

\begin{proof}
  Part~\ref{enum-req-tx}) is implied directly by
  \cref{lem:sinr-lb-and-cr-tx}. Regarding Part~\ref{enum-receive-ai}),
  it holds that $v$'s leader node $w$ must answer at most~$\Delta$
  requests. For each request $w$ transmits the node's answer for
  $\kappa_2$ time slots. Thus $v$ receives its active interval at
  most~$\kappa_1 + \Delta\kappa_2$ time slots after it started
  transmitting its request.  Part~\ref{enum-ai-waittime}) holds by
  definition of $\tau(\cdot,\cdot)$.
\end{proof}

We shall now argue that each non-leader node succeeds to win a second
level MIS in its active interval.

\begin{lemma}
  \label{lem:async-schedule-node-in-mis}
  Given a node $v$ executing \cref{algo:level-2}. Once $t=0$, $v$
  wins a second level MIS set within $\slai$ time slots.
\end{lemma}

\begin{proof}
  The active interval of $v$ is from $t_s = 0$ to $t_e = \slai$.
  According to Lemma~\ref{lem:async-schedule-const-neighbors}, at most
  $\slan$ nodes are active in $v$'s broadcasting range at any time
  (including $v$). As nodes are active for $\slai$ consecutive time
  slots, it holds that during the interval $\{t_s,\dots,t_e\}$ at most
  $2\slan-2$ neighbors of $v$ are active as well. With each MIS
  execution at the second level at least one neighbor gets to be in an
  independent set, selects a valid color, transmits this color to its
  neighbors and does not participate in following MIS
  executions. Thus, after at most $2\slan-2$ MIS executions $v$ either
  succeeded in a previous MIS execution, or all active neighbors of
  $v$ succeeded and thus $v$ must succeed in the next second level MIS
  execution.
\end{proof}

As a final step we show that the final color the nodes select are
valid w.h.p.
\begin{lemma}
  \label{lem:async-colors-valid}
  Given a node $v$ entering \cref{algo:Colored-l}. It holds that   
  \begin{inparaenum}[a)]
  \item \label{enum-tx-no-confl} while $v$ transmits its final color
    no neighbor of $v$ succeeds in a second level MIS w.h.p.; and
  \item \label{enum-color-valid} the color $v$ selects is not selected
    by one of $v$'s neighbors w.h.p.
  \end{inparaenum}
\end{lemma}

\begin{proof}
  Part~\ref{enum-tx-no-confl}) holds since all neighbors of $v$ that
  participated in the current second level MIS know the counter value
  of $v$ w.h.p., and hence do not enter the MIS before $v$ finished
  transmitting its final color. As the nodes are not synchronized, a
  node $w$ may have just entered the active interval. In this case,
  however, the listening phase of $\kappa_2$ time slots prevents~$w$
  from succeeding before it knows the status of all active neighbors.
  Thus, while $v$ transmits its color, no neighbors of $v$ select their
  final color (except for first level leaders taking color 0, which
  does not conflict with the second level coloring).  For
  Part~\ref{enum-color-valid}) observe that the algorithm removes each
  final color it receives from the set of available colors according
  to \cref{algo:async-color-reduction}. Due to
  Part~\ref{enum-tx-no-confl}) it holds that the colors of all
  neighbors of $v$ that selected a final color before $v$ were able to
  transfer the color to~$v$ successfully. Thus $v$ selected a color
  which was not selected by one of its neighbors.
\end{proof}

We are now able to prove the main theorem. Note that runtime bounds
hold for each node beginning with the node start executing the
algorithm.
\begin{proof}[Proof of \cref{thm:async-correct-runtime}]
  The number of colors used follows directly from the algorithm. It
  follows from \cref{lem:async-colors-valid} and the fact that each
  node succeeds in an MIS (and hence enters \cref{algo:Colored-l} and
  selects a final color), that the final color of each node is valid
  w.h.p. A simple union bound over all nodes implies that the
  $\Delta+1$ coloring is valid.  The runtime follows as the first
  level MIS takes $\O(\Delta\log n)$ time slots according to
  \cref{lem:mis-independence-runtime}, \cref{algo:level-2} takes
  another $\O(\Delta \log n)$ slots until starting the active
  interval, which is of length $\O(\log n)$.
\end{proof}

\begin{corollary}
  Let each node in the network execute the MW-coloring algorithm from
  \cite{dt-dncsi-10} followed by
  Algorithm~\ref{algo:async-color-reduction}. Then $\O(\Delta \log n)$
  time slots after a node started executing the algorithm it has a
  valid color from $[\Delta]$.
\end{corollary}

\vspace{2mm}
{
\noindent\textbf{Concurrent Execution of Algorithms: }
\label{sec:conc-exec-algor}
We consider this algorithm isolated from other algorithms that may be
executed simultaneously by other nodes (which just woke up, for
example). The additional interference introduced by a constant number
of different algorithms that are executed simultaneously in the
network can be handled simply by reducing the transmission
probabilities used in the algorithms by the number of algorithms
executed simultaneously.


\section{Conclussion}
\label{sec:conclussion}

We conclude that the proposed distributed $\Delta+1$ coloring
algorithms are simple and very fast. They are based on a simple and
well-known randomized $\Delta+1$ coloring algorithm and a color
reduction scheme.  By carefully balancing the uncertainties due to
communication in the physical or SINR model and the uncertainties in
the randomized algorithm itself, we are able to guarantee a runtime of
$\O(\Delta \log n)$ time slots for all our algorithms. This
corresponds to just a few rounds of local broadcasting and outperforms
all previous algorithms either by the number of colors required or the
running time.

\vspace{2mm}
{
\noindent\textbf{Acknowledgements: }
\label{sec:acknowledgements}
We thank Magn\'{u}s M.\ Halld\'{o}rsson for helpful discussions on an
early stage of this work, and the German Research Foundation (DFG),
which supported this work within the Research Training Group GRK 1194
"Self-organizing Sensor-Actuator Networks".  }


\bibliographystyle{IEEEtran}
\bibliography{abbrv,bib}

\begin{appendix}

\section{Appendix}

\subsection{Omitted Proofs}
\label{sec:omitted-proofs}

Restatement of Lemma~\ref{lem:sum-tx-prob}, which establishes a bound
on the sum of transmission probabilities from within each broadcasting
region based on the transmission probabilities $p_1$ and $p_2$ used
for communication in the algorithms.
\lemmasumtxbounded*

\begin{proof}
  Depending on the algorithm, nodes in $v$'s broadcasting range are
  either transmitting with probability $p_1$ or $p_2$.  It holds that
  most nodes transmit with probability $p_1$ for all algorithms. For
  the use of probability $p_2$ in Algorithm~\ref{algo:colorreduction}
  is holds that at most a 5 nodes from within each broadcasting range
  transmit with probability $p_1$, since in a broadcasting range in
  which more than 5 nodes have the same color, two of them must be
  neighbors. This would violate the validity of the $4\Delta$ coloring
  of the network, which is guaranteed by
  Proposition~\ref{prop:randcolor-valid-whp}.  For
  \cref{algo:async-color-reduction} it holds that at most 90 nodes
  within $v$'s broadcasting range use $p_1$ according to
  \cref{lem:async-schedule-const-neighbors}.

  Let us now consider both cases jointly, and let $p_w$ be the current
  transmission probability of $w$, which is either $p_1$, $p_2$ or 0
  if $w$ is currently not trying to transmit (e.g. if $c^t_v \not = i$
  in Algorithm~\ref{algo:colorreduction}).  Let us now bound the sum
  of transmission probabilities from within $B_v$.
  \begin{align*}
    \sum_{w \in B_v,\atop w \text{ transmits}} p_w &\leq \sum_{w \in B_v}
    p_1 + \sum_{w \in B_v,\atop w \text{ transmits with }p_2} p_2 \\ 
    & \leq \Delta \cdot
    \frac{1}{2\Delta^A} + 90 \cdot \frac{1}{180} \leq 1,
  \end{align*}
  where the last inequality holds since $r_A > 2r_B$ which
  implies~$\Delta \leq \Delta^A$.
\end{proof}

\subsection{Successful Transmission with Constant Probability}
\label{sec:succ-transm-with}

In this section, we give a full proof for Theorem~\ref{thm:sinr}. The
following proof is largely based on the proof of Lemma 4.1 and Theorem
4.2 by Goussevskaia, Moscibroda and Wattenhofer
\cite{gmw-lbpim-08}. Before proving the result, let us introduce some
additional notation required in this section.

\subsubsection{Further definitions}
\label{app:sinr:further-definitions}

Definitions of the transmission and the broadcasting range are given
in Section~\ref{sec:prelim-coloring}. We shall give more details about
the proximity range here. Let us consider an arbitrary node $v$. In
accordance with \cite{gmw-lbpim-08}, we define the proximity range
around $v$ to be $r_A = r_B \left(3^3 2^\alpha \beta \cdot
  \left(\frac{\alpha-1}{\alpha-2}\right) \right) ^
{\frac{1}{\alpha-2}}$, and denote set of nodes within the transmission
range of $v$ by $A_v$.
  Let $\chi := \frac{2 \pi}{3 \sqrt{3}}
  \frac{(r_A+2r_B)^2}{r_B^2}$. This is, intuitively speaking, bound on
  the number of independent broadcasting ranges that can be fit in a
  disk of radius $r_A$.

\subsubsection{Full proof of Theorem~\ref{thm:sinr}}
\label{app:sinr:proof-of-theorem}

We shall now prove Theorem~\ref{thm:sinr} (restated in the following).
\sinrthm*

\begin{proof}
  Let us first prove the two major observations that are implied by
  the fact that the sum of transmission probabilities from within each
  broadcasting range is bounded by 1.  These observations are
  formulated in the two following claims. The probabilities proven in
  these claims can be combined with the transmission probability to
  form the probability that $v$ succeeds in transmitting its message
  in a given time slot. After proving the claims, we shall show that
  these probabilities indeed imply a constant success probability of
  $11/12$. 

  \begin{claim}
    The probability $P^{A_v}_\text{none}$ that $v$ is the only node in
    the proximity region that transmits a signal in the current time
    slot is at least $(1/4)^\chi$, and thus constant.
  \end{claim}
  \begin{proof}[Proof of the claim]
    Let us consider the probability that any other node in $v$'s
    proximity region attempts to transmit. Again, let $p_w$ be the current
    transmission probability of $w$, which is either $p_a$, $p_s$ or 0
    if $w$ is currently not trying to transmit. 
    \begin{align*}
      P^{A_v}_\text{none} & \geq \prod_{w \in A_v\backslash\{v\}} (1-
      p_w) \geq \left( \frac{1}{4} \right)^{\sum_{w \in
          A_v\backslash\{v\}} p_w} \\
      & \geq \left( \frac{1}{4} \right)^{\sum_{u\in
          A_v\backslash\{v\},\atop u \text{ independent}}\sum_{w \in
          B_u}
        p_w} \\
      & \geq \left( \frac{1}{4} \right)^{\sum_{u\in
          A_v\backslash\{v\},\atop u \text{ independent}} 1} \geq
      \left( \frac{1}{4} \right)^\chi,
    \end{align*} 
    where the second inequality holds due to Fact~3.1 in
    \cite{gmw-lbpim-08}, the third inequality follows by covering the
    nodes in $A_v$ by broadcasting ranges of independent nodes in
    $A_v$. The next inequality is implied by
    Lemma~\ref{lem:sum-tx-prob}, while the last inequality holds as
    $\chi$ is (roughly speaking) the number of independent nodes in
    $A_v$.  The claim follows from observing that $\chi$ is indeed
    constant as the number of disks required to cover a (by a
    constant) larger disk is constant (cf. Fact 3.3 in
    \cite{gmw-lbpim-08}).
  \end{proof}
  As the proof of the following claim is in parts exactly as already
  covered in \cite{gmw-lbpim-08}, we omit the bound on the
  interference received from nodes outside of the proximity area of
  $v$ on nodes in the broadcasting range of $v$.
  \begin{claim}
    The probability $P^v_\text{SINR}$ that the SINR constraint
    holds for a given transmission is at least $1/2$.
  \end{claim}
  \begin{proof}[Proof of the claim] The proof is based on the concept
    of rings around the transmitting node $v$. With increasing
    distance the number of nodes in a ring increases, however, also
    the effects on nodes in the broadcasting range of $v$ decreases.
    Based on Lemma~\ref{lem:sum-tx-prob}, the bound on the
    interference received at an arbitrary node $w$ in $v$'s
    broadcasting range can be bound by $\frac{P}{4\beta r_B^{\alpha}}$
    exactly as in \cite{gmw-lbpim-08}. Now, applying the Markov
    inequality it holds that the probability that the interference is
    more than twice this level (and thus the SINR constraint is
    violated) is at most $1/2$. It follows that the SINR constraint
    holds with probability at least $1/2$.
  \end{proof}

  Combining the probabilities from the previous claims with the
  transmission probability $p_v$, it follows that the probability that $v$
  successfully transmits a message to all neighbors in a given time
  slot is at least $p_v \cdot P^{A_v}_\text{none} \cdot
  P^v_\text{SINR}$. As $\lambda := \left( P^{A_v}_\text{none} \cdot
    P^v_\text{SINR} \right)^{-1}$, the probability for a successful
  transmission to all neighbors after $\frac{\lambda \ln 12}{p}$
  time slots is at least
\begin{align*}
  1- \left( 1 - \frac{p}{\lambda} \right)^{\frac{\lambda \ln 12}{p} }
  \geq 1- e^{-\ln 12} \geq 1- \frac{1}{12} \geq \frac{11}{12},
\end{align*}
which concludes the proof.
\end{proof}

\subsection{Discussion and Future Work}
 \label{sec:discussion}

   {\bf Local Synchronization:} Let us briefly discuss the assumptions
   of our algorithms for the synchronous setting. A wireless network
   can be synchronized either by using network-wide broadcasts, or by
   applying synchronization methods such as Timing-sync Protocol for
   Sensor Networks (TPSN) \cite{gks-tpsn-03}. Assuming a node with a
   valid clock within few hops of each node, the network can be
   synchronized within a few rounds of local broadcasting. Then, the
   simple synchronous coloring algorithms presented in
   \cref{sec:algorithm} can be used to compute a $\Delta+1$
   coloring.

{\bf Without knowledge of $\Delta$:}
\label{sec:without-knowledge-delta}
If this upper bound on the local density is not known, one must begin
with a very low transmission probability and use a so-called
slow-start technique to find the correct transmission
probability. Current slow-start techniques for local broadcasting
require $\O(\Delta \log n + \log^2 n)$ time slots
\cite{hm-ttblb-12,yhwl-aodaa-12}. Unfortunately, we cannot adapt the
slow-start technique to achieve successful message transmission with
constant probability as shown for the case with knowledge of $\Delta$
in Section~\ref{sec:sinr-model-related}.

For the algorithm by Halld\'{o}rsson and Mitra \cite{hm-ttblb-12},
reducing the runtime by a $\log n$ factor by allowing success with
only constant probability (instead of w.h.p.) fails as we cannot
guarantee that the so-called fallback occurs with high probability,
which is required to keep the sum of local transmission probabilities
low. It would be interesting to find an algorithm that achieves
message transmission with constant probability without the knowledge
of $\Delta$ that improves the runtime bound of
$\O(\Delta \log n + \log^2 n)$.  Such an algorithm could generalize
the algorithms presented in this work to operate without prior
knowledge of $\Delta$.

\end{appendix}


\end{document}